\documentclass[conference]{IEEEtran}
\usepackage{amsmath,amssymb,epsfig,psfrag,cite,subfigure}
\include{macros}
\usepackage{graphicx}
\usepackage{bbm}

\usepackage{algorithm}
\usepackage{subfigure}
\usepackage{color}

\usepackage{tikz}
\usepackage{pict2e}
\usepackage{scthm}
\usepackage{amsmath,amsfonts,amsthm} % Math packages
\usepackage{commath}
\usepackage{pgfplots}
\usepackage{bbm}

\newtheorem{theorem}{Theorem}
\newtheorem{lemma}{Lemma}

\newtheorem{corollary}{Corollary}
\newtheorem{definition}{Definition}

\newtheorem{example}{Example}
\newtheorem{property}{Property}

\newcommand{\bs}[1]{\boldsymbol{#1}}

\newcommand{\E}[1]{\mathbb{E}\left[#1\right]}

\def \cO {\mathcal{O}}

\def \cX {\mathcal{X}}
\def \cY {\mathcal{Y}}
\def \bF {\mathbb{F}}
\def \cQ {\mathcal{Q}}
\def \cR {\mathcal{R}}
\def \cA {\mathcal{A}}
\def \bP {\mathbb{P}}
\def \cB {\mathcal{B}}
\def \cG {\mathcal{G}}
\def \cF {\mathcal{F}}

\renewcommand{\baselinestretch}{0.93}

\begin{document}

\title{Finite-Level Quantization Procedures for Construction and Decoding of Polar Codes}
\author{\authorblockN{Yunus Inan and Emre Telatar}
	\authorblockA{EPFL, Lausanne, Switzerland}{Email: \{yunus.inan,emre.telatar\}}@epfl.ch}
 
 \maketitle 
\begin{abstract}
	We consider finite-level, symmetric quantization procedures for construction and decoding of polar codes. Whether polarization occurs in the presence of quantization is not known in general. In \cite{Urbanke}, it is shown that a simple three-level quantization procedure polarizes and a calculation method is proposed to obtain a lower bound for achievable rates. We find an improved calculation method for achievable rates and also the exact asymptotic behavior of the block error probability under the aforementioned simple case. We then prove that certain $D$-level quantization schemes polarize and we give a lower bound on achievable rates. Furthermore, we show that a broad class of quantization procedures result in a weaker form of the polarization phenomenon.
\end{abstract}

\section{Introduction}
\vspace{0.2cm}

Polar codes are the first class of channel codes that achieve capacity for Binary-input Memoryless Symmetric (BMS) channels with low encoding and decoding complexities \cite{Arikan_polarization}. As the name suggests, polar codes are based on a polarization phenomenon, which we now describe briefly: Given two identical and independent instances of a BMS channel $W: \bF_2 = \cX \to \cY$, create two synthetic channels $W^-: \cX \to \cY^2$ and $W^+:\cX \to \cY^2 \times \cX$ with the polar transform introduced in \cite{Arikan_polarization}. Arikan has shown that the mutual information of $W^+$ is greater than the mutual information of $W^-$ and their average is equal to that of $W$. This means that from a BMS channel $W$, its `worse' and `better' versions are synthesized while the average mutual information is preserved. Recursive application of the above construction allows one to synthesize channels $W^{\bs{s}_n}$ for all $\bs{s}_n \in \{+,-\}^n$ in $n$ steps. Arikan has also shown that a fraction of synthetic channels eventually become `perfect' whereas the other fraction eventually become `useless'. In other words, they eventually polarize. Together with the fact that the average mutual information remains same at each step and the error probability of perfect channels behave as $O(2^{-2^{n/2}})$ (cf. \cite{Telatar_Rate}), this shows the capacity achieving property of polar codes.

 Arikan has introduced the Successive Cancellation Decoder (SCD) in \cite{Arikan_polarization}, which estimates the channel input sequence by calculating the individual log-likelihood ratios (LLR) for each bit, exploiting the recursive structure. The basis of code construction is to send the information bits through synthetic channels that are close to perfect. Identifying these almost perfect channels can in principle be done with a density evolution algorithm \cite{MoriTanaka}. We exploit the inherent symmetry of BMS channels and assume all-zero sequence is sent throughout this manuscript. Under this assumption and supposing that the random channel output is $Y$, the update equations for LLRs are given by
\begin{equation}\label{eq:minus_llr}
L^- = L \boxplus L', \qquad L^+ = L + L'
\end{equation}
where $L\triangleq \ln\left(\frac{W(Y|0)}{W(Y|1)}\right)$, $a\boxplus b \triangleq \ln\left(\frac{e^{a+b}+1}{e^a + e^b}\right)$ and $L'$ is an identical and independent copy of $L$. Similar to the creation of synthetic channels, one can calculate the distribution of any $L^{\bs{s}_n}$, $\bs{s}_n \in \{+,-\}^n$. Note that the distribution of $L^{\bs{s}_n}$ is equivalent to the channel transition probabilities of $W^{\bs{s}_n}$ given all-zero input.

Now, we state two challenges about code construction and decoder implementation: 
\begin{enumerate}
	\item In general, equations \eqref{eq:minus_llr} suggest that the support size of LLRs grow exponentially in block length. To overcome this problem, special degradation procedures or approximations are proposed (e.g., see \cite{Tal_Vardy}, \cite{RamtinHassani}).
	\item LLRs are real numbers, therefore implementation of a real-time SCD has to include an inherent quantization scheme depending on the required precision (c.f. \cite{LLRBasedCircuit}). In \cite{Urbanke}, robustness of polarization with respect to a specific family of quantization schemes was examined and the authors have shown that even a simple 3-level quantization scheme polarizes.
\end{enumerate}

We refer the reader to the partial list (\!\!\cite{Trifonov,MoriTanaka2,UniformQuantization,Joachim,HigherOrder}) for other studies on these considerations. To the best of our knowledge, little is known about polarization phenomenon for finite-level quantization schemes other than a specific three-level case. We have found that a weaker polarization phenomenon compared to that in \cite{Arikan_polarization} exists under some constraints. 

The main results of this manuscript are:

\begin{itemize}
	\item[(i)] For the three-level quantization scheme in \cite{Urbanke}, an improved calculation method for the lower bound for achievable rates is obtained.
	\item[(ii)] The exact asymptotic behavior of block error probability for the same three-quantized decoder is found to be $\cO(2^{-\sqrt{N}^{\log \phi}})$, where $\phi = \frac{1+\sqrt{5}}{2}$ is the golden ratio and $N = 2^n$ is the block length.
	\item[(iii)] A broad family of finite-level quantization procedures weakly polarize. The family is to be defined in Section \ref{sec:Qfamily}.
	\end{itemize}

\section{Notation}

The random variables are denoted with uppercase letters whereas their realizations are denoted with lowercase letters (e.g., $X_n$ and $x_n$). Sets and events are denoted with script-style letters (e.g., $\cA_n$, $\cG_n$). As two special cases, the set $\{1,2,\ldots,n\}$ is denoted $[n]$, $n\in \mathbb{N}$ and $\Pi_{\mathbb{R}}$ denotes the set of all probability distributions on $\mathbb{R}$.  $|\cA|$ denotes the cardinality of a set $\cA$. Vectors and sequences are denoted by boldface letters. If their length is known, it is added as a subscript (e.g., $\bs{s}_n$). If the length is not known or has no importance, we drop the subscript (e.g., $\bs{s}$). $\mathbbm1_{\cA}$ denotes the indicator function for a set $\cA$.

We abbreviate the following operations: $a \wedge b \triangleq \min\{a,b\}$, $a \vee b \triangleq \max\{a,b\}$, $\text{sign}(x) \triangleq \mathbbm{1}_{\{x > 0\}} - \mathbbm{1}_{\{x < 0\}}$. $h(x) \triangleq -x\log x - (1-x)\log(1-x)$ is the binary entropy function defined for $x \in [0,1]$. All the logarithms are in base 2 unless we use the notation $\ln$ for natural logarithm.

\section{Static and Dynamic Quantization Procedures}\label{sec:Qfamily}
We define a family of symmetric quantization procedures to unify the approaches in \cite{Urbanke},\cite{Tal_Vardy} and \cite{RamtinHassani}.
\begin{definition}[$D$-quantization family and admissible quantization procedures] \label{def:qfamily} For a finite $D \in \mathbb{N}$, a $D$-quantization family $\cQ^{(D)}$ is a family of odd, increasing step functions which can take at most $D$ values. Moreover, the members are right continuous on $\mathbb{R}_+$, and left continuous on $\mathbb{R}_-$. We also define the family of admissible quantization procedures as $\cQ \triangleq \cup_{D \geq 1} \cQ^{(D)}$.
\end{definition}
Restriction to odd functions provides symmetry. This is necessary to preserve the property that the set of BMS channels are invariant under polar transforms with quantization schemes.

Note that Definition \ref{def:qfamily} implies that for all $Q \in \cQ$, $Q(0)=0$. Hence, one can always take $D$ as an odd number. Furthermore, for any member of $\cQ$; the quantization intervals in $\mathbb{R}_+$ together with their images contain all the information needed for its behavior in $\mathbb{R}$. Taking into account the above, we have the following definition of static and dynamic quantization procedures.
\begin{definition}[$D$-static and $D$-dynamic quantization]\label{def:quantization}
	A $D$-dynamic quantization $Q_{\beta}^{(D)}:\Pi_{\mathbb{R}}\times\mathbb{R} \to \mathbb{R}$ is a member of $\cQ^{(D)}$, where the right limits of quantization intervals in $\mathbb{R}_+$ and their images are described in parameter $\beta(\bP)$, $\bP \in \Pi_{\mathbb{R}}$. $\beta(\bP)$ is a set of 2-tuples with $|\beta| = \frac{D-1}{2}$ and depends on the distribution $\bP$. A $D$-static quantization is a $D$-dynamic quantization with $\beta$ being same for all $\bP \in \Pi_{\mathbb{R}}$.
\end{definition}

We give a simple example of a $D$-static quantization procedure.
\begin{example}
	Given $\alpha_1, \alpha_2, \gamma_1, \gamma_2 \in \mathbb{R},~0 < \alpha_1 < \alpha_2$ and $0 < \gamma_1 < \gamma_2$, let $\beta = \{(\alpha_1,\gamma_1),(\alpha_2,\gamma_2)\}$. $Q_\beta^{(5)}(x)$ is depicted in Figure \ref{fig:Q}:
\begin{figure}[h!]\label{fig:Q}\centering
	
	\begin{tikzpicture}[scale=0.82, every node/.style={scale=0.85}]
	\begin{axis}[
	axis x line=center,
	axis y line=center,
	x axis line style={<->},
	y axis line style={<->},
	xtick={-0.7,-0.2,0,0.2,0.7},
	ytick={-1,-0.6,0,0.6,1},
	xmajorgrids,
	ymajorgrids,
	yticklabels={,,},
	xticklabels={,,},
	ymin=-2,ymax=2,
	]
	\addplot[mark = * ,blue,line width = 0.8pt, mark options={fill=white}] coordinates {
		(-0.7,-1)(-0.7,-0.6) (-0.2,-0.6) (-0.2,0) (0.2,0) (0.2,0.6) (0.7,0.6)(0.7,1)
	};
\addplot[mark = none,blue,line width = 0.8pt] coordinates {
	(0.7,1) (1,1)
};
\addplot[mark = none,blue,line width = 0.8pt] coordinates {
	(-0.7,-1) (-1,-1)
};
\addplot[only marks, mark = *, blue] coordinates {
	(-0.7,-1) (-0.2,-0.6)  (0.2,0.6) (0.7,1)
};

\node at (axis cs: 0.7,-0.2) {$\alpha_2$};
\node at (axis cs: 0.2,-0.2) {$\alpha_1$};
\node at (axis cs: -0.7,0.2) {$-\alpha_2$};
\node at (axis cs: -0.2,0.2) {$-\alpha_1$};
\node at (axis cs: -0.1,0.6) {$\gamma_1$};
\node at (axis cs: -0.1,1) {$\gamma_2$};
\node at (axis cs: 0.1,-0.6) {$-\gamma_1$};
\node at (axis cs: 0.1,-1) {$-\gamma_2$};
\node at (axis cs: 0.17,1.8) {$Q_\beta^{(5)}(x)$};
\node at (axis cs: 0.94,-0.1) {$x$};
	\end{axis}

	\end{tikzpicture}
	\caption{Graphical representation of $Q_\beta^{(5)}(x)$.}
\end{figure}
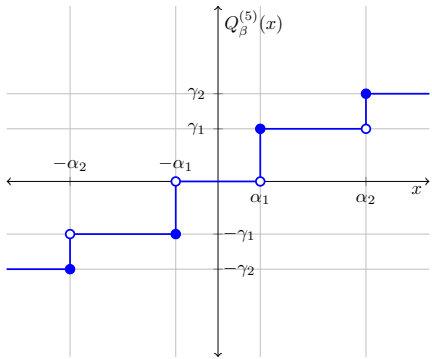

A special case is when $\alpha_1 = 0$. Then, $Q_\beta(0) = 0$ and $Q_\beta(x) = \gamma_1$ for $0<x<\alpha_2$. Observe that $Q_\beta$ is not continuous at zero for this case.
\end{example}

We sometimes drop the superscript $(D)$ if the number of quantization levels $D$ is known or trivial. For dynamic quantization procedures, the notation $\beta(Y)$ is equivalent to $\beta(\bP)$ if a random variable $Y$ with distribution $\bP$ is to be quantized.

$\cQ$ contains a broad class of practical quantization procedures. Observe that any quantization scheme similar to those in \cite{Urbanke} belongs to $\cQ^{(D)}$. Furthermore, it is immediate from Definition \ref{def:quantization} that $Q \circ Q' \in \cQ$ for all $Q$, $Q'\in \cQ$. This implies that the greedy quantization procedures in \cite{Tal_Vardy} and \cite{RamtinHassani} are dynamic quantization procedures which belong to $\cQ$ with the additional condition that zero is an absorbing support, namely, any combination of the zero support with some nonzero support should map to zero. We also emphasize that the widely used approximation (c.f. \cite{HardwarePolar})
$$
a\,  \widehat \boxplus\, b \triangleq (|a|\wedge|b|)\text{sign}(ab) \approx a \boxplus b
$$
results in a dynamic quantization procedure under some conditions.

\begin{lemma}\label{lem:minsum}Consider a discrete random variable $L$ and its identical and independent copy $L'$ that take values in the finite set $\mathcal{L} = \{d_1, \ldots, d_n\}$ for some $n \in \mathbb{N}$. Take the symmetrized set $\tilde{ \mathcal{L}} \triangleq \mathcal{L}\, \cup\, (-\mathcal{L})$, where $-\mathcal{L} = \{-d_1, \ldots, -d_n\}$. Suppose the non-negative elements of $\tilde{ \mathcal{L}}$ are ordered as $\alpha_1\leq\ldots\leq\alpha_m$ for some  $m$. If $\alpha_{i+1} > \ln(e^{\alpha_i} + \sqrt{e^{2\alpha_i}-1})$ for all $1 \leq i \leq m-1$, there exists a dynamic quantization procedure $Q_{\beta(L)}$ such that $L\, \widehat \boxplus\, L' = Q_{\beta(L)}(L \boxplus L')$.
\end{lemma}
\begin{proof} The random variable $L \boxplus L'$ takes values in the set $\tilde{ \mathcal{L}}\,\boxplus\, \tilde{ \mathcal{L}} =  \{-(\alpha_m \boxplus \alpha_m), \ldots, \alpha_m \boxplus \alpha_m\}$. Suppose $\alpha_{i+1} > \ln(e^{\alpha_i} + \sqrt{e^{2\alpha_i}-1})$ for all $1 \leq i \leq m-1$, then one can show $\alpha_{i-1} < \alpha_{i} \boxplus \alpha_{i} < \alpha_{i} \boxplus \alpha_{i+1} < \ldots < \alpha_{i} \boxplus \alpha_m < \alpha_{i}$ for all $i \in [m]$. Take the dynamic quantization procedure $Q_{\beta{(L)}}$ with 
	$$
	\beta(L) = \cup_{i = 1}^m\{(\alpha_{i} \boxplus \alpha_{i},\alpha_{i})\}.
	$$
	With the above selection, $Q_{\beta(L)}\bigl(\cup_{j \geq i}^m\{\alpha_{i} \boxplus \alpha_{j}\}\bigr) = \alpha_i = \cup_{j \geq i}^m\{\alpha_{i} \wedge \alpha_{j}\}$ for all $i \in [m]$. In other words, every $\alpha_i \boxplus \alpha_j$ is mapped to $\alpha_i \wedge \alpha_j$. Since this true for all $i,j \in [m]$, $\bigl(L\, \widehat \boxplus\, L'\bigr) \vee 0 = \bigl(Q_{\beta(L)}(L \boxplus L')\bigr) \vee 0$. The proof for the negative support follows similarly.
\end{proof}

Note that the condition in Lemma \ref{lem:minsum} can be met by simply scaling the random variables $L$, $L'$ with a large enough constant.
%We remark an important observation at this point. It is known that LLRs are sufficient statistics. When there is no quantization, LLRs of synthetic channels can be calculated with equations in \eqref{eq:minus_llr} and the input bits are estimated via LLRs without any information loss; which gives the impression that one should proceed with LLRs even in the presence of quantization. This may not be the optimal way in general. It is possible that under certain quantization procedures, some other statistics preserve more information than LLRs.
 
\section{Three-Quantized Case}\label{sec:3quantizedCase}
In this section, we study the same three-level quantization procedure from \cite{Urbanke}. We briefly explain the findings in \cite{Urbanke} with an improvement on calculation of the lower bound for the fraction of perfect channels. We also find the exact asymptotic behavior of the block error probability.

Consider a BMS channel $W$, whose output $Y$ takes values from the set $\{-\lambda,0,\lambda\}$. If the initial channel has support size larger than three, it can be quantized with any desired procedure until we obtain a channel with three outputs.
% More clearly, $W : \bF_2 \to \{-\lambda,0,\lambda\}$ for some $\lambda \geq 0$. For instance, if one works with LLRs, $\lambda$ and channel transition probabilities are related with $\lambda = \log\left(\frac{W(\lambda|0)}{W(\lambda|1)}\right) = \log\left(\frac{W(\lambda|0)}{W(-\lambda|0)}\right)$. Furthermore, we assume without loss of generality $W(\lambda|0) \leq W(\lambda|1)$. Otherwise, one inverts the channel to fulfil this condition. 
The static quantization procedure we consider throughout this section is $Q_\beta^{(3)}$, $\beta = \{(0, 1)\}$. Verbally, $Q_\beta$ results in only propagating the signs of the quantized random variables. The quantized channel output, $Y^{\bs{s}_n}= Q_\beta(Y^{\bs{s}_{n-1},s_n})$, $\bs{s}_n \in \{+,-\}^n,$ $n\geq 1$ with $Y^{\bs{s}_{n-1},s_n}$ defined according to \eqref{eq:minus_llr}; has therefore three parameters, namely $p^{\bs{s}_n} \triangleq \Pr(Y^{\bs{s}_n}=1)$, $z^{\bs{s}_n} \triangleq \Pr(Y^{\bs{s}_n}=0)$ and $m^{\bs{s}_n} \triangleq \Pr(Y^{\bs{s}_n}=-1)$. Without loss of generality, we assume $p\geq m$. Otherwise, one can negate the channel output  to fulfil this condition. These parameters completely describe the distribution of $Y^{\bs{s}_n}$. Referring to \eqref{eq:minus_llr}, iterations of $(p,m,z)$ under $Q_\beta$ are given by
\vspace{-0.1cm}
\begin{equation}\label{eqn:updates3}
\vspace{-0.1cm}
\begin{split}
p^{+}= p^2 + 2pz & \qquad p^{-} = p^2 + m^2\\
m^+ = m^2 + 2mz &  \qquad m^-= 2mp\\
z^+  = z^2 + 2mp &  \qquad z^-= 2z-z^2\text.
\end{split}
\end{equation}
These iterations are the same as those in \cite{Urbanke}. It is possible to calculate $(p^{\bs{s}},m^{\bs{s}},z^{\bs{s}})$ for any ${\bs{s}} \in \{+,-\}^*$ with the above transformations. Note that these transformations preserve $p^{\bs{s}} \geq m^{\bs{s}}$.

\subsection{Feasible Region for $Y^{\bs{s}}$}
Our purpose is to track these parameters for the statistic $Y^{\bs{s}}$. At first sight, it may seem that $p^{\bs{s}}$, $m^{\bs{s}}$ and $z^{\bs{s}}$ can take any value in the set $\cR_3 \triangleq \{(p,m,z):p\geq m,\ p+m+z = 1,\ p,m,z \geq 0\}$. However, this is not the case. If it is known that $Y^{\bs{s}}$ has gone through $+$ transformation once, there are some restrictions on the feasible region for its parameters.
\begin{lemma}\label{lem:region}
	Define the limiting curve as the $(p,m)$ pairs with the following parametric equations:
	\begin{equation}
	\begin{split}\label{eqn:limiting_curve}
	p^*(t) &= \sqrt{4t^3-3t^4}\\
	m^*(t)& = 1 - 3t + \frac 3 2 t^2 + \frac {p^*(t)} 2\text, \quad t \in [0,1]\text.
	\end{split}
	\end{equation}
	Let $\cR_3^+ \triangleq  \cR_3 \cap \{0 \leq m \leq m^*(t),\ p = p^*(t),\ \forall t \in [0,1] \}$. Then, for any $\bs{s}_n \in \{+,-\}^n$, $n\geq 1$
	\begin{itemize}
		\item[(i)] It is sufficient that $\bs{s}_n$ contains at least one $(+)$ to ensure that $(p^{\bs{s}_n},m^{\bs{s}_n}) \in \cR_3^+$.
		\item[(ii)] If $(p^{\bs{s}_n},m^{\bs{s}_n}) \in \cR_3^+$, then $(p^{\bs{s}_ns},m^{\bs{s}_ns}) \in \cR_3^+$ for $s \in \{+,-\}$. In words, once $(p^{\bs{s}_n},m^{\bs{s}_n})$ is driven under the limiting curve, it remains there.
		\end{itemize}
\end{lemma}

Proof of Lemma \ref{lem:region} is given in Appendix \ref{pf:region}.

\subsection{Polarization of Quantized Statistics}\label{sec:prob_setting}

With a similar approach to those in \cite{Urbanke} and  \cite{Arikan_polarization}, parameters of quantized statistics can be examined in a probabilistic setting. The setting is described below:
\\
Fix $\Omega \triangleq \{+,-\}^*$ and let $\bs{S}_n = (S_1,S_2,\ldots,S_n)$ be a sequence of $n$ random variables where each $S_i$ is independently and uniformly distributed on $\{+,-\}$. Define the natural filtration $\{\mathcal{F}_n\}_{\mathbb{N}}$ with $\mathcal{F}_n \triangleq \sigma(\bs{S}_n)$, $n \geq 1$ and $\mathcal{F}_0 \triangleq \{\Omega,\emptyset\}$. Also define $\mathcal{F} \triangleq \sigma((S_n)_\mathbb{N})$. These ingredients completely define the probability space with filtration $(\Omega,\mathcal{F},\{\mathcal{F}_n\},\mathbb{P})$ and for a quantized statistic obtained in $n$ polarization steps, any of its parameter becomes an $\cF_{n}$-measurable random variable, namely $P_n \triangleq p^{\bs{S_n}}$, $Z_n \triangleq z^{\bs{S_n}}$ and $M_n \triangleq m^{\bs{S_n}}$. Also note that any function of $D_n \triangleq (P_n,Z_n,M_n)$ becomes random.

The quantized statistic $Y^{\bs{S}_n}$ can also be represented as a 'quantized' or 'degraded' synthetic BMS channel $\tilde W^{\bs{S}_n}$ with
$$
\tilde {W}^{\bs{S}_n}(y|0) = \begin{cases} P_n, & y = 1\\
Z_n, & y = 0\\
M_n, & y = -1
\end{cases}\text.
$$

It is known that any bounded submartingale or supermartingale converges almost surely (see, e.g. \cite{Martingales}). Therefore, if a function of $D_n$ is a submartingale or supermartingale, it may give information on whether polarization occurs. From this perspective, we list some consequences of the quantization procedure $Q_\beta$ in terms of probabilistic arguments. One can verify that $P_n$, $M_n$, $Z_n$ themselves exhibit submartingale/supermartingale properties \cite{Urbanke}. Moreover, the mutual information of $\tilde{W}^{\bs{s}_n}$, 
	$$
	I(\tilde{W}^{\bs{s}_n}) \triangleq (p^{\bs{s}_n}+m^{\bs{s}_n})\left(1-h\left(\frac{m^{\bs{s}_n}}{p^{\bs{s}_n}+m^{\bs{s}_n}}\right)\right)
	$$
	is a supermartingale. This property follows simply from data processing inequality as the average mutual information is preserved without quantization. 
\begin{lemma}[\!\cite{Urbanke}, Lemma 4]\label{lem:polarize}
	The random variables $P_n$, $Z_n$, $M_n$ converge almost surely. Moreover, $Z_\infty \triangleq \lim_{n \to \infty} Z_n = 0$ or $1$,
	$P_\infty \triangleq \lim_{n \to \infty} P_n = 0$ or 1 and 	$M_\infty \triangleq \lim_{n \to \infty} M_n = 0$ almost surely. Namely, $Y^{\bs{S}_n}$ polarizes.
\end{lemma}
Lemma \ref{lem:polarize} simply follows from the fact that $Z_n$ is a submartingale and $M_n$ supermartingale.

Knowing that the quantized statistics polarize, we elaborate on the question of what fraction of these statistics carry lossless information. We note that it is very hard to obtain an exact expression for this fraction. Let $\gamma$ denote the fraction of the lossless statistics. Lower and upper bounds for $\gamma$ can be obtained from the submartingale and supermartingale properties of some functions $f(D_n)$ with $f(1,0,0) = 1$ and $f(0,1,0) = 0$. Suppose $f(D_n)$ is a bounded submartingale (supermartingale), i.e., it satisfies $\frac{f(d^+)+f(d^-)} 2  \substack{\geq \\ (\leq)}  f(d), \forall d \in \cR_3$. Then $\gamma = \E {f(P_\infty,Z_\infty,M_\infty)} \substack{\geq \\ (\leq)} f(p,z,m)$, which shows that $f$ is useful to obtain an lower (upper) bound on $\gamma$. In \cite{Urbanke} it is shown that $I(W)^2 \leq \gamma \leq I(W)$ as $I(\tilde{W}^{\bs{S}_n})$ is a supermartingale and $I(\tilde{W}^{\bs{S}_n})^2$ submartingale.
In addition, we have numerically found that $I^{1.24}(\tilde W^{\bs{S}_n})$ is submartingale if the process starts in $\cR_3^+$. Hence, we have the following improved lower bound for $\gamma$.
\begin{lemma}\label{lem:new_bound}
	If the original $(p,m)$ belongs to $ \cR_3^+$, then $I^{1.24}(W)$ is a lower bound for $\gamma$. If not, then $\frac 1 2 {I^{1.24}(\tilde W^+)} + \frac 1 2 I^2(\tilde W^-)$ is a lower bound for $\gamma$. More precisely, define 
	\begin{equation*}
	F_0(W) \triangleq \begin{cases} I^{1.24}(W), & (p,m) \in R_3^+\\
	\frac 1 2 I^{1.24}(\tilde W^+) + \frac 1 2 I^2(\tilde W^-),& \text{else}
	\end{cases}\text.
	\end{equation*}
	Then, $F_0(W) \leq \gamma$. 
	\end{lemma}
\begin{corollary}\label{corr:new_bound_n}$F$ can be improved by increasing the number of polarization steps. Namely, define
	$$
	F_n(W) \triangleq \begin{cases} \frac 1 {2^n} \sum_{\bs{s}_n \in \{+,-\}^n}I^{1.24}(\tilde W^{\bs{s}_n}),\\ \qquad(p,m) \in R_3^+\\
	 \frac 1 {2^n}\left(\sum_{\bs{s}_n \in \{+,-\}^n \setminus (-)^n}I^{1.24}(\tilde W^{\bs{s}_n}) + I^2(\tilde W^{(-)^n}\!)\!\right)\!\!,\\ \qquad \text{else}
	\end{cases}\text.
	$$
	Then, $F_0(W) \leq F_n(W) \leq \gamma$.
	\end{corollary}
The proposed method for calculation of the lower bound in \cite{Urbanke} relies on the fact that $\gamma$ is bounded from above and below as $\E{I^{2}(\tilde W^{\bs{S}_n})} \leq \gamma \leq \E{I(\tilde W^{\bs{S}_n})}$, and $\E{I(\tilde W^{\bs{S}_n})} - \E{I^2(\tilde W^{\bs{S}_n})} \leq \delta$ for some $\delta >0$ and large enough $n$. Therefore, one can obtain a confidence interval of $\delta$ for large $n$. Since $ \E{I(\tilde W^{\bs{S}_n})}-F_n(W)$ decreases faster, the same confidence interval $\delta$ can be achieved with smaller $n$ compared to the first method. This results in an improved calculation method for the lower bound.
\subsection{Rate of Polarization}
From the previous section, we know that the quantized statistics polarize. However, it is required that the error probability $P_e(\tilde W^{\bs{S}_n}) \triangleq M_n + \frac 1 2 Z_n$ of each perfect statistic decays fast enough, i.e. $o(2^{-n})$, to ensure reliable communication under the aforementioned quantization procedure. For the unquantized case, it is found in \cite{Telatar_Rate} that the Bhattacharyya parameter $Z_b(W^{\bs{S}_n})$, which is an upper bound to the error probability, decays as $O(2^{-2^{n/2}})$ and in \cite{Urbanke}, it is shown that $Z_b(\tilde W^{\bs{S}_n}) \triangleq 2\sqrt{P_nM_n} + Z_n$ decays as $O(2^{-2^{\alpha n}})$, $\alpha < \frac {\log 1.5} 2$ under $Q_\beta$ according to the previously given probabilistic setting. Since $(P_n,M_n) \in \cR_3^+$ and thus $M_n \leq Z_n$ eventually, this also implies $Z_n$ and $M_n$ decay at least with the same rate. However, one cannot compare the decay rates of $Z_n$ and $M_n$ only knowing the decay rate of $Z_b(\tilde W^{\bs{S}_n})$. If $M_n$ decays much faster than $Z_n$, it is possible that the code constructed with $Q_\beta$ can be concatenated with an erasure-only code as an outer code for large $n$. Unfortunately, this is not the case. To show this, we present the following lemma and theorem, whose proofs are given in Appendices \ref{pf:process_R} and \ref{pf:rate} respectively.
\begin{lemma}\label{lem:processR}
	For all $\epsilon_r > 0$, 
	$$\lim_{n\to \infty} \bP \left(\left|\frac{\log M_n}{\log Z_n} - \phi\right| \leq \epsilon_r \right)= \gamma\text,\quad\phi = \frac{1+\sqrt 5}{2}\text.$$
	\end{lemma}

Lemma \ref{lem:processR} suggests that with probability close to $\gamma$, $M_n$ and $Z_n$ decay with same rate. With the next theorem, we obtain the exact rate.

\begin{theorem}\label{thm:rate}
	In limit, the random processes $Z_n$ and $M_n$ roughly behave as $O(2^{-2^{\alpha n}})$, $\alpha = \frac {\log \phi}{2}$ with probability close to $\gamma$. That is, for any $\delta,\delta'> 0$,
			\vspace{-0.2cm}
	$$
	\lim_{n\to \infty} \bP\left(2^{-2^{n\frac{\log \phi + \delta'}{2}}} \leq Z_n \leq 2^{-2^{n\frac{\log \phi - \delta}{2}}}\right)= \gamma
			\vspace{-0.2cm}
	$$
	\centerline{and}
	\vspace{-0.2cm}
	$$
\lim_{n\to \infty} \bP\left(2^{-2^{n\frac{\log \phi + \delta'}{2}}} \leq M_n \leq 2^{-2^{n\frac{\log \phi - \delta}{2}}}\right)= \gamma\text.
	$$
	\end{theorem}
Lemma \ref{lem:processR} and Theorem \ref{thm:rate} imply that $Z_n$ and $M_n$ decay at the same rate. Consequently, concatenation with an erasure-only code does not improve the error probability. Also note that the rate of polarization for this particular three-quantized case is bounded away from $O(2^{-2^{n/2}})$, which shows that longer codes are required to ensure reliable communication compared to the unquantized case.

\section{D-Quantized Case}
In this section, we consider static and dynamic quantization procedures $Q_{\beta}^{(D)}$, where $D = 2d+1$ is an odd number by definition. Note that $|\beta| = d $. Similar to the three-level case, we start with a BMS channel $W$ whose output $Y$ takes values in the set $\{0,\pm \lambda_1,\ldots,\pm \lambda_d\}$, $\lambda_i > 0$, $i \in [d]$. Define the parameters of the quantized statistic $Y^{\bs{s}_n}$ as $p_i^{\bs{s}_n}$ , $m_i^{\bs{s}_n}$ and $z^{\bs{s}_n}$ in a similar fashion to that in Section \ref{sec:3quantizedCase} and assume $p_i \geq m_i$. Also define $p^{\bs{s}_n} \triangleq \sum_{i=1}^d p^{\bs{s}_n}_i$ and $m^{\bs{s}_n}\triangleq \sum_{i=1}^d m^{\bs{s}_n}_i$.

In general, it appears to be hard to obtain good lower bounds on the achievable rates for quantization procedures with output size greater than three. However, we have found that there are non-trivial $D$-static and $D$-dynamic quantization procedures that result in the same dynamics as the simple three-quantized case. We formally define these procedures below.
\begin{definition}[Proper quantization procedures]
	A quantization procedure $Q_{\beta(\bP)}^{(D)}$ is proper if $\beta(\bP)_i \neq \beta(\bP)_j$ for all $i\neq j \in [d]$ and $\bP \in \mathcal{P}$. In words, $\beta$ consists of distinct elements.
	\end{definition}
Note that if a quantization procedure is not proper, then it is equivalent to another quantization procedure with $|\beta| < d$.
\begin{lemma}\label{lem:properQ}
	There exists 
	\begin{itemize}
		\item[(i)] a pair of proper $D$-static quantization procedures $Q_{\beta^+}$, $Q_{\beta^-}$ with $Y^+ = Q_{\beta^+}(Y + Y')$, $Y^- = Q_{\beta^-}(Y \boxplus Y')$ that results in the same dynamics as the three-quantized case,
		\item[(ii)] a single proper $D$-static quantization procedure $Q_\beta$ that results in the same dynamics as the three-quantized case.
		\end{itemize}
	\end{lemma}
\begin{proof}[Proof Sketch] 
	\begin{itemize}
		\item[(i)]
		Take any 
	$$\beta^+ = \cup_{i=1}^d\{(\alpha_i, \alpha_i)\}
	,\quad \beta^- = \cup_{i=1}^d\{(\alpha_i\boxplus \alpha_i, \alpha_i)\}$$
	 such that $0 < \alpha_1 <\alpha_i < 2\alpha_1$, $i \in [d]$, $i\neq 1$.
	 \item[(ii)] Take $\beta = \cup_{i=1}^d\{(\alpha_i \boxplus \alpha_i, \alpha_i)\}$ such that $0 < \alpha_1 <\alpha_i < 2(\alpha_1 \boxplus \alpha_1)$, $i \in [d]$, $i\neq 1$.
	 \end{itemize}
 Under these assumptions, one can verify that the resulting dynamics for both cases become the same as those in the formerly discussed three-quantized case.
\end{proof}
Lemma \ref{lem:properQ} shows that with a pair of two proper $D$-static quantization procedures, or with a single proper $D$-static quantization procedure, the system performance can be made equivalent to that in the simple three-quantized case. This also implies that there are proper $D$-dynamic quantization schemes with the same performance. Based on this fact, a lower bound on the achievable rates can be derived for $D$-quantization families.
\begin{lemma}\label{lem:Dlower_bnd}
	Consider the function $F_n$ defined in Corollary \ref{corr:new_bound_n} for an $n \geq 0$. Then, the following claims hold:
	\begin{itemize}
		\item[(i)] With a pair of proper $D$-static quantization procedures $Q_{\beta^+}$ and $Q_{\beta^-}$, one can achieve rates greater than
		$$
		R_{s,2}^{(D)}(W) \triangleq \max_{\stackrel{\alpha_1 \leq \alpha_2 \ldots \leq \alpha_d}{\stackrel{\alpha_d \leq 2\alpha_1}{\alpha_1 \boxplus \alpha_1\vee(\alpha_1/2) \leq \lambda_d}}} \frac{F_n(\tilde W^+) + F_n(\tilde W^-)}{2}\text,
		$$
		where 
		$\beta^+ = \cup_{i=1}^d\{(\alpha_i, \alpha_i)\}
		$
		and
		$\beta^- = \cup_{i=1}^d\{(\alpha_i\boxplus \alpha_i, \alpha_i)\}\text.
		$
		\item[(ii)] With a single proper $D$-static quantization procedure $Q_{\beta}$, one can achieve rates greater than
		$$
		R_{s,1}^{(D)}(W) \triangleq \max_{\stackrel{\alpha_1 \leq \alpha_2 \ldots \leq \alpha_d}{\stackrel{\alpha_d \leq 2(\alpha_1 \boxplus \alpha_1)}{\alpha_1 \boxplus \alpha_1\leq \lambda_d}}} \frac{F_n(\tilde W^+) + F_n(\tilde W^-)}{2}\text,
		$$
		where $\beta = \cup_{i=1}^d\{(\alpha_i\boxplus \alpha_i, \alpha_i)\}$.
		\item[(iii)] With a proper $D$-dynamic quantization procedure $Q_\beta$, one can achieve rates greater than
		$$
				R_d^{(D)}(W) \triangleq  \sup_{\stackrel{Q_{\beta(Y+Y')}\in \cQ^{(D)}}{Q_{\beta(Y\boxplus Y')}\in \cQ^{(D)}}} \frac{F_n(\tilde W^+) + F_n(\tilde W^-)}{2}\text,
		$$
		where $Y^+ = Q_{\beta(Y+Y')}(Y+Y')$ and $Y^- = Q_{\beta(Y\boxplus Y')}(Y+Y')$. In other words, quantize $Y+Y'$ and $Y\boxplus Y'$ in the best possible way to maximize the objective function. 
		\end{itemize}
	\end{lemma}
\begin{proof} For (i) and (ii), take the procedures described in Lemma \ref{lem:properQ}. Since the evolution of the parameters are same as the three-quantized case after one polarization step, we use the same lower bound. The last inequalities are added to make the region compact. For (iii), we see that at any step, a proper dynamic quantization exists to ensure that the parameters evolve similarly to the three-quantized case. Quantization at first step is optimized to get a better lower bound.
	\end{proof}
	It is important to note that the special quantization schemes considered in the proof of Lemma \ref{lem:properQ} ensure that the quantized statistics polarize as the resulting dynamics are equivalent to that in three-level case. At first glance, it is not obvious that the statistics polarize for any admissible quantization procedure. Surprisingly, the quantized statistics polarize in a weaker manner under any admissible static or dynamic quantization procedure.
	\begin{theorem}\label{thm:Dpolarize} Consider the probabilistic setting in Section \ref{sec:prob_setting} and define $P_{i,n} \triangleq p_i^{\bs{S}_n}$, $M_{i,n} \triangleq m_i^{\bs{S}_n}$ for all $i \in [d]$. Then, for all static or dynamic quantization procedures in $\cQ$, $Z_n$ converges to $0$ or $1$ almost surely and for any $i$, $P_{i,n}M_{i,n}$ converges to $0$ in probability.
		\end{theorem}
	\begin{proof} We use the abbreviations $X_n \stackrel{\text{a.s.}}{\to} c$ and $X_n \stackrel{P}{\to} c$ to denote that $X_n$ converges to $c\in \mathbb{R}$ almost surely or in probability respectively. For every static or dynamic $Q_\beta \in \cQ$, it is known that $Q_\beta(0) = 0$. This implies that if $Y = 0$ or $Y'=0$ then $Y^- = Q_\beta(Y \boxplus Y') = 0$ and if $Y,Y' = 0$ or $Y = -Y'$ then $Y^+ =  Q_\beta(Y + Y') = 0$. One thus obtains
		\vspace{-0.1cm}
		$$
		 z^- \geq 2z-z^2,\quad z^+ \geq z^2 + 2\sum_{i = 1}^d p_im_i\text.		\vspace{-0.1cm}
		$$
		
		Therefore, $Z_n$ is a bounded submartingale as $\E{Z_{n+1}|\cF_{n}} \geq Z_n +  \sum_{i = 1}^d P_{i,n}M_{i,n}$. Considering the $-$ transformation and following the same steps in \cite{Arikan_polarization}, we obtain
		\begin{align*}
		\E{|Z_{n+1} - Z_n|} &\geq \frac 1 2\E{Z_n^- - Z_n} \geq \frac 1 2\E{Z_n-Z_n^2}\text.
		\end{align*}
		Since $\lim_n \E{|Z_{n+1} - Z_n|}  = 0$ and $Z_n$ converges almost surely, $Z_n \stackrel{\text{a.s.}}{\to} 0\text{ or }1$.
		Studying the $+$ transformation instead, we obtain
		\begin{align*}
		\E{|Z_n^+ - Z_n|} = \E{\left|Z_n^2-Z_n + 2\sum_{i=0}^d P_{i,n}M_{i,n} + J_n \right|}\text,
		\end{align*}
	
		where $J_n$ is an $\cF_{n}$-measurable non-negative remainder term. With a similar reasoning, we know that the right hand side goes to zero as $n$ tends to infinity. This implies that $Z_n^2-Z_n + 2\sum_{i=0}^d P_{i,n}M_{i,n} + J_n\stackrel{P}{\to} 0 $.  $Z_n^2-Z_n\stackrel{\text{a.s.}}{\to} 0$ implies $Z_n^2-Z_n \stackrel{P}{\to} 0$. It is well-known that if $X_n \stackrel{P}{\to} x$ and $Y_n \stackrel{P}{\to} y$ for some constants $x$ and $y$, then $X_n+Y_n \stackrel{P}{\to} x+y$. From this fact, we conclude that $2\sum_{i=0}^d P_{i,n}M_{i,n} + J_n \stackrel{P}{\to} 0$ as well. Since both $2\sum_{i=0}^d P_{i,n}M_{i,n} $ and $J_n$ are non-negative random variables, we have $\sum_{i=0}^d P_{i,n}M_{i,n} \stackrel{P}{\to} 0$ and $P_{i,n}M_{i,n} \stackrel{P}{\to} 0$ for all $i \in [d]$.
		\end{proof}

 Theorem \ref{thm:Dpolarize} has significance in practice as it implies Tal-Vardy construction in \cite{Tal_Vardy} under the assumption that zero is an absorbing support, any quantization scheme as in \cite{Urbanke} and many other schemes weakly polarize. The weak polarization implies that for sufficiently large $n$, some fraction of synthetic channels meet the condition that $\tilde W^{\bs{s}_n}(y|0)$ and $\tilde W^{\bs{s}_n}(y|1)$ have almost non-overlapping supports.  If one is allowed to remap the supports and change the quantization procedure once at some $n$, one can show that the quantized statistics can be forced to polarize strongly.
 
 \begin{lemma}\label{lem:FinalPolarize} Assume $Z_\infty = 0$ with probability $\gamma_Z > 0$, i.e., a non-zero fraction $\gamma_Z$ of quantized statistics tend to become non-zero with probability 1. Given $\epsilon,\delta >0$ and $\delta \leq \gamma_Z$, one can ensure that the quantized statistics polarize and at least $(\gamma_Z-\delta)(1-\epsilon-2\sqrt{d}\epsilon^{1/4})^2$ fraction of the statistics will eventually become perfect by remapping of supports and changing the procedure to the simple three-quantized case after some $n_0(\delta,\epsilon)$.
 	
 	\end{lemma}
\begin{proof} Given $\epsilon,\delta$, Theorem \ref{thm:Dpolarize} implies the existence of an $n_0$ such that
	$$
	\bP(Z_n \leq  \epsilon, P_{i,n}M_{i,n} \leq \epsilon,\ i \in [d]) \geq \gamma_Z -\delta,\quad n \geq n_0\text.
	$$
	We consider $\bs{s}_n \in \{+,-\}^n$ such that the condition in the above event holds. For such $\bs{s}_n$, $p^{\bs{s}_n}_i \wedge m^{\bs{s}_n}_i \leq \sqrt{\epsilon}$ for all $i \in [d]$. 
	At $n_0$, we remap the support such that $m^{\bs{s}_n}_i \gets p^{\bs{s}_n}_i \wedge m^{\bs{s}_n}_i$ and we switch to the simple three-level quantization procedure $Q_\beta$, $\beta = \{(0,1)\}$. This will ensure that $m^{\bs{s}_n} \leq d\sqrt{\epsilon}$. Under these conditions the Bhattacharyya parameters are bounded as $Z_b(\tilde W^{\bs{s}_n}) \triangleq z^{\bs{s}_n} + 2\sqrt{p^{\bs{s}_n}m^{\bs{s}_n}} \leq \epsilon + 2\sqrt{d}\epsilon^{1/4}$. For BMS channels, it is known that $I(W) \geq 1-Z_b(W)$, thus $I(\tilde W^{\bs{s}_n}) \geq 1-\epsilon - 2\sqrt{d}\epsilon^{1/4}$. Observe that the specific three-quantized case polarizes strongly. Now we use the simple lower bound $I(W)^2$ to show that at least $(\gamma_Z-\delta)(1-\epsilon-2\sqrt{d}\epsilon^{1/4})^2$ fraction of channels will eventually become perfect.
	\end{proof}

Note that the three-quantized case assures that the block error probability behaves roughly as $O(2^{-\sqrt{N}^{\log\phi}})$. Together with Lemma \ref{lem:FinalPolarize}, it implies that one achieves reliable communication at rates arbitrarily close to $\gamma_Z$ by constructing and decoding polar codes with $D$-level quantization procedures, if it is allowed to change the procedure and remap the supports once at an arbitrary $n$. As a final note, we remark that if the quantization procedures take some special form, e.g., if they ensure that the quantized statistics are LLRs as in \cite{Tal_Vardy}, then the remapping of the support is not needed since $M_n \leq P_n$ always.

%
%\begin{lemma} The set of functions $m_\beta(p) = (1-p^{1/\beta})^\beta$, $\beta \geq \log 3$, provide lower bounds to the limiting curve.
%\end{lemma}
%\begin{proof} Since $m_{\beta}(p) \geq m_\alpha(p)$ for $\alpha \geq \beta \geq \log3$ and for all $p \in [0,1]$, it suffices to prove the lemma for $\beta = \log 3$.
%\end{proof}
\bibliographystyle{IEEEtran}
\bibliography{Ref4}

\appendix

\subsection{Proof of \lemref{lem:region}}\label{pf:region}
\begin{itemize}
	\item[(i)] Our purpose here is to show that when $\bs{s}_n$ contains at least one $(+)$, $(p^{\bs{s}_n},m^{\bs{s}_n})$ is driven under the limiting curve. In other words, for a fixed $p^{\bs{s}_n+}$, we want to prove that $m^{\bs{s}_n+}$ cannot exceed the limiting curve. To this end, using \eqref{eqn:updates3}, we formulate the following optimization problem.
	\begin{align*}
	\underset{p,m}{\max}\quad & m^+ = 2m-m^2-2mp\\
	\text{s.t}\quad & p^+ = 2p-p^2-2mp\\
	& p,m \geq 0\\
	& p+m \leq 1
	\end{align*}
	where $p^+$ is a fixed constant in $[0,1]$. From the equality constraint, we have $m = (1-p/2-p^+/2p)$ and the objective function can be modified as
	\begin{gather*}
	2m-m^2+p^2-2p = 1-\left(\frac p 2 + \frac {p^+} {2p}\right)^2 +p^2-2p\\
	= (p-1)^2 + \left(\frac p 2 + \frac {p^+} {2p}\right)^2\text.
	\end{gather*}
	Taking the derivative and setting to 0, we obtain the only extremal $p$ in an implicit function
	$$
	p^+ = \sqrt{4p^3-3p^4}\text.
	$$
	The same extremal $p$ yields the maximized objective function
	$$
	m^+ = 1-3p+\frac 3 2 p + \frac 1 2 \sqrt{4p^3-3p^4}\text.
	$$
	Note that the map $p \to p^+$ is bijective in $[0,1]$. This gives a parametric description of $(p^+,m^+)$, where $p^+,m^+ \in [0,1]$ for $p \in [0,1]$. However, we note that for $p \leq 1/3$, $m^+ \geq p^+$ which is a contradiction to our assumptions. Incorporating the fact that $p^+$ is always greater than $m^+$, the parametric curve can be described as above for $p \in [1/3,1]$, and $p^+ = m^+ = p$ for $p \in [0,1/3)$. Renaming the variable $p$ as $t$, we obtain the same parametric description given in the statement of \lemref{lem:region}.
	
	The optimization problem above was formulated to find the maximum $m^+$ value corresponding to a $p^+$. Hence, given $\bs{s}_n$ contains at least one $(+)$, we have shown that $(p^{\bs{s}_n},m^{\bs{s}_n})$ cannot exceed the limiting curve and any such $(p^{\bs{s}_n},m^{\bs{s}_n})$ is driven into $\cR_3^+$.
	
	\noindent Before proving part (ii), we give the following property.
	\begin{property}\label{prop:limiting}
		The limiting curve is non-increasing and convex on $p \in [1/3,1]$. Moreover, $\frac{\partial m^*}{\partial p^*} \geq -1$ and $\frac{\partial^2m^*}{\partial (p^*)^2}\vert_{p^* \to 1} = \infty$.
	\end{property}
	\begin{proof}
		$\frac{\partial m^*}{\partial p^*} = \frac 1 2 \left(1-\sqrt{\frac 4 t -3}\right)$.
		$$
		\frac{\partial^2m^*}{\partial (p^*)^2} = \frac{\frac{\partial}{\partial t} \frac{\partial m^*}{\partial p^*} }{\frac{\partial p^*}{\partial t} } = \frac 1 {6(t^2-t^3)} \geq 0.
		$$
		The inequality and limit argument follows easily.
	\end{proof}

	\item[(ii)] For this part, we have to show that once a $(p^{\bs{s}_n},m^{\bs{s}_n})$ is driven under the limiting curve, it remains there. Similar to part (i), we consider the following optimization problem to find the maximum value of a $m^-$ with respect to a fixed $p^-$:
	\begin{align*}
	\underset{p,m}{\max}\quad & m^- = 2mp\\
	\text{s.t}\quad & p^- = p^2+m^2\\
	& (p,m) \in \cR_3^+
	\end{align*}
	where $p^-$ is a fixed constant in $[0,1]$. It is easy to see that the optimal $(\hat p,\hat m)$ for this problem also maximizes the function $p+m$. Therefore, for $p^- \leq \frac{2}{9}$, $\hat p = \hat m = \sqrt{\frac{p^-}{2}}$ and $p^- = m^- = 2\hat m \hat p$. The $(p^-,m^-)$ corresponding to $(\hat p,\hat m)$ remains in $\cR_3^+$. If $p^- > \frac{2}{9}$, then the optimal $(\hat p,\hat m)$ always lies on the limiting curve. Therefore, the parametric description for the solution is given by
	\begin{equation*}
	\begin{split}
	\tilde{p}(t) &= (p^*(t))^2 + (m^*(t))^2\\
	& = (4t^3-3t^4) + \left(1-3t+\frac 3 2 t^2 + \frac{\sqrt{4t^3-3t^4}}{2}\right)^2\\
	\tilde{m}(t) &= 2p^*(t)m^*(t)\\
	&=\sqrt{4t^3-3t^4}\left(1-3t+\frac 3 2 t^2\right) + \frac 1 2 (4t^3-3t^4)\text.
	\end{split}
	\end{equation*}
	for $t \in [1/3,1]$.\\
	
		Now, one has to check if $(\tilde{p}(t), \tilde{m}(t)) \in \cR_3^+$ for all $t \in [1/3,1]$. Observe that for any $(p,m) \in \cR_3^+$, $p^- = p^2 + m^2 \leq p$, thus $\tilde p(t) \leq p^*(t)$. The equality holds if and only if $t = 1$. Moreover, we note that $\tilde m(p)$ has to be convex in $p > p_c$ for some critical $p_c$ as its derivative is zero at $p=1$ and being concave will drive it to the negative side, which is impossible. From these facts, we observe that if $\tilde m(p)$ exceeds $m^*(p)$ at some $p$, it is required that $\frac{\partial m^*}{\partial p'} \geq \frac{\partial \tilde m}{\partial p'}$ for some other $p' \geq p$.
		 Hence if we show that this inequality does not hold, then the proof will be complete. Noting that $\tilde p(t) \leq p^*(t)$, it is sufficient to prove the stronger statement
	\begin{equation}\label{eq:derivative}
	\frac{\frac{\partial\tilde{m}}{\partial t}}{\frac{\partial\tilde{p}}{\partial t}} \geq 	\frac{\frac{\partial m^*}{\partial t}}{\frac{\partial p^*}{\partial t}},\quad t \in [1/3,1]\text.
	\end{equation}
	One can derive 
	\begin{equation*}
	\begin{split}
	\frac{\partial\tilde{p}}{\partial t} = 2p^*(t)\frac{\partial p^*}{\partial t} + 2m^*(t)\frac{\partial m^*}{\partial t}\text,\\
	\frac{\partial\tilde{m}}{\partial t} = 2p^*(t)\frac{\partial m^*}{\partial t} + 2m^*(t)\frac{\partial p^*}{\partial t}\text.
	\end{split}
	\end{equation*}
	Hence, 
	$$
	\frac{\frac{\partial\tilde{m}}{\partial t}}{\frac{\partial\tilde{p}}{\partial t}} = \frac{2p^*(t)\frac{\partial m^*}{\partial t} + 2m^*(t)\frac{\partial p^*}{\partial t}}{2p^*(t)\frac{\partial p^*}{\partial t} + 2m^*(t)\frac{\partial m^*}{\partial t}} = \frac{\frac{m^*(t)}{p^*(t)} + \frac{\frac{\partial m^*}{\partial t}}{\frac{\partial p^*}{\partial t}}}{1+\frac{m^*(t)}{p^*(t)}\frac{\frac{\partial m^*}{\partial t}}{\frac{\partial p^*}{\partial t}}}\text.
	$$
	The inequality \eqref{eq:derivative} then becomes
	$$
	\frac{\frac{m^*(t)}{p^*(t)} + \frac{\frac{\partial m^*}{\partial t}}{\frac{\partial p^*}{\partial t}}}{1+\frac{m^*(t)}{p^*(t)}\frac{\frac{\partial m^*}{\partial t}}{\frac{\partial p^*}{\partial t}}} \geq \frac{\frac{\partial m^*}{\partial t}}{\frac{\partial p^*}{\partial t}}
	$$
	and if the denominator is positive for all $t \in [1/3,1]$, we have
	$$
	1 \geq \left(\frac{\frac{\partial m^*}{\partial t}}{\frac{\partial p^*}{\partial t}}\right)^2,
	$$
	which is correct regarding Property \ref{prop:limiting}. 
	As the final step, we show that the denominator is positive. First, note that $\frac{\partial p^*}{\partial t} \geq 0$ and $\frac{\partial m^*}{\partial t} \leq 0$. Then,
	\begin{equation*}
	\begin{split}
	1+\frac{m^*(t)}{p^*(t)}\frac{\frac{\partial m^*}{\partial t}}{\frac{\partial p^*}{\partial t}} \geq 1 + \frac{\frac{\partial m^*}{\partial t}}{\frac{\partial p^*}{\partial t}} \geq 0\text,
	\end{split}
	\end{equation*}
	which is again satisfied because of Property \ref{prop:limiting}, and the first inequality follows from the fact that $m^* \leq p^*$.\\
	
	These together prove that for any $(p,m) \in \cR_3^+$, $(p^-,m^-)$ lies under the limiting curve and hence belongs to $\cR_3^+$. It straightforwardly follows from part (i) that $(p^+,m^+)$ also belongs to $\cR_3^+$. Therefore once a pair $(p,m)$ is driven into $\cR_3^+$, it remains there.
	\end{itemize}
\subsection{Proof of \lemref{lem:processR}} \label{pf:process_R}
To begin with, the following upper bound for the limiting curve will be useful for the proof.

\begin{lemma}\label{lem:upper_bnd}
	The curve $\bar m(p) = C(1-p)^{3/2},\ C\geq 2$ lies above the limiting curve.
\end{lemma}
\begin{proof}
	According to the parametric description \eqref{eqn:limiting_curve}, choose any $t \geq 1/3$. At this $t$, we have 
	\begin{align*}
	p^*(t) = \sqrt{4t^3-3t^4}\\
	m^*(t) = 1-3t+\frac 3 2 t^2 + \frac{\sqrt{4t^3-3t^4}}{2}\text.
	\end{align*}
	For the chosen $t$, $\bar m(t) = C(1-\sqrt{4t^3-3t^4})^{3/2}$. Now, one needs to check if
	$$
	C(1-\sqrt{4t^3-3t^4})^{3/2} \geq 1-3t+\frac 3 2 t^2 + \frac{\sqrt{4t^3-3t^4}}{2}.
	$$
	We use the upper bound $2t-t^2 \geq \sqrt{4t^3-3t^4}$ to obtain the stronger statement
	\begin{equation}
	\begin{split}\label{eq:ineq_upper}
	C(1-(2t-t^2))^{3/2} \geq 1-3t+\frac 3 2 t^2 + \frac{\sqrt{4t^3-3t^4}}{2}\\
	\iff C(1-t)^3 \geq 1-3t+\frac 3 2 t^2 + \frac{\sqrt{4t^3-3t^4}}{2}
	\end{split}
	\end{equation}
	With a change of variable $v \triangleq 1-t$ and rearranging the terms, we have
	$$
	C\geq \frac{\frac 3 2 v^2 -\frac 1 2 + \frac{\sqrt{4(1-v)^3-3(1-v)^4}}{2}}{v^3} \triangleq g(v)\text.
	$$
	Observe that 
	$$
	\lim_{v \to 0} g(v) = 2,\quad g(1) = 1\text,
	$$
	hence $g$ is bounded in $(0,1]$. Therefore, if one takes $C = \sup_{v \in (0,1)}g(v)$, the inequality \eqref{eq:ineq_upper} is satisfied. We now show that $g(v)$ is decreasing in $(0,1)$. Taking the derivative, we have
	$$
	g'(v) = \frac{3\sqrt{1-v}\left((1-v)^2 -2 + (1+v)\sqrt{(1-v)(1+3v)}\right)}{2v^4\sqrt{(1+3v)}}\text.
	$$
	It suffices to check if the nominator is non-positive in $(0,1)$. To this end, we need to verify the following statement.
	$$
	h(v) \triangleq (1-v)^2 -2 + (1+v)\sqrt{(1-v)(1+3v)} \leq 0\text.
	$$
	To find the extrema of $h(v)$ in $(0,1)$, we take the derivative of $h(v)$ and equate to zero. 
	\begin{align*}
	h'(v) &= -2(1-v) + \sqrt{(1-v)(1+3v)}\\
	&  + \frac{(1+v)(1-3v)}{\sqrt{(1-v)(1+3v)}} = 0\\
	\iff &-2(1-v)\sqrt{(1-v)(1+3v)}\\
	&+ (1-v)(1+3v) + (1+v)(1-3v) = 0\\
	\iff &1-3v^2 - (1-v)\sqrt{(1-v)(1+3v)} = 0\\
	\iff &1-3v^2 = (1-v)\sqrt{(1-v)(1+3v)}\\
	\iff &1-6v^2 + 9v^4 = (1-v)^3(1+3v),\quad \text{for }v < \frac{1}{\sqrt 3}\\
	\iff& 4v^3(3v-2) = 0\\
	\iff &v = 2/3\text. 
	\end{align*}
	However, $\frac 2 3 > \frac 1 {\sqrt 3}$. Therefore, $h(v)$ has no extremal points in $(0,1)$. Observe that $h$ is continuous and $h(0) = 0$, $h(1) = -1$. These together imply $h(v) < 0$ for $v \in (0,1)$. Hence we have shown that $g'(v) < 0$ for $v \in (0,1)$ and $g(v)$ is decreasing on the same interval. Finally, we obtain $\sup_{v \in (0,1)}g(v) = 2$.\\
\end{proof}

We are now in position to prove \lemref{lem:processR}.

Let $C_n \triangleq \frac{Z_n^2}{P_nM_n}$. Choose a $\delta$ such that $\delta < \epsilon_r$ and $\delta(2+\delta)e^\delta - \frac 4{\log \delta} \leq 1/2$ (e.g. $\delta < 0.003$). Choose a small $\epsilon > 0$.

 Now, define the event $\cA_n(\delta) \triangleq \{P_n  \geq 1-\delta\}$. From the almost sure convergence of $P_n$, we know that
	$$
	\bP(\cup_m \cap_{n\geq m} \cA_n(\delta))  = \lim_m \bP( \cap_{n\geq m} \cA_n(\delta)) = \gamma\text.
	$$
	The sequence above is increasing. Hence, given $\epsilon$, there exists an $n_0(\delta,\epsilon)$ such that
	$$
	\bP(\cap_{n\geq n_0} \cA_n(\delta)) \geq  \gamma-\epsilon/3.
	$$
	This also implies that $\bP(\cap_{k = n_0}^n \cA_k(\delta)) \geq  \gamma-\epsilon/3$ for any $n \geq n_0$. Define $\cB_{n,m}(\delta) \triangleq \cap_{k = m}^n \cA_k(\delta)$. For any $\bs{s}_n \in \cB_{n,n_0}(\delta)$, $n\geq n_0$, the iterations for $C_{n+1}$ can be upper bounded as below. We drop the subscripts and use lowercase characters for simplicity.
	\begin{align*}
	c^+ & = \frac{(z^2 	+ 2mp)^2}{mp(m+2z)(p+2z)} = \frac{mp(z^2/mp + 2)^2}{(m+2z)(p+2z)}\\
	& \leq \frac{mp(z^2/mp + 2)^2}{mp+4z^2+2z\left(m+p\right)} = \frac{(c+2)^2}{3+4c}\\
	&\leq \begin{cases}c, & c > \frac 4 3\\
	\frac 4 3, & c \leq \frac 4 3\end{cases}\text.\\
	c^- &= \frac{(2z-z^2)^2}{2mp(p^2+m^2)} = \frac{c(1+m+p)^2}{2(p^2+m^2)} \leq c\left(1+\frac{1}{m+p}\right)^2\\
	&\leq c\left(\frac{2-\delta}{1-\delta}\right)^2 \leq 9c
	\end{align*}
	since $\delta < 1/2$.
	We create another process $D_n$ as follows: Let $C^*(\delta,\epsilon) = C_{n_0}^* \triangleq \max_{\bs{s}_{n_0}\in\{+,-\}^{n_0}} c^{\bs{s}_{n_0}} \vee \frac 4 3 $. Then,
	$$
	D_{n+1} = 9D_n,\quad n\geq n_0,
	$$
	$$
	D_{n_0} = C_{n_0}^*\text.
	$$
	It is easy to see that for any $\bs{s}_n \in \cB_{n,n_0}(\delta)$, $n\geq n_0$; $D_n$ dominates $C_n$ and therefore,
	\begin{equation}\label{eq:process_C}
	C_n \leq  C^*(\delta,\epsilon)9^{n-n_0}\text.
	\end{equation}

Let $A_n \triangleq -\log M_n$, $B_n \triangleq -\log Z_n$. For $\bs{s}_n \in \cB_{n,n_0}(\delta)$, $n\geq n_0$, we derive upper and lower bounds for $a^+$, $a^-$ and $b^+$, $b^-$:
\begin{equation}\label{eq:abupdates}
	\begin{gathered}
	a-1  \leq a^- \leq a,\\
	b-1 \leq b^- \leq b,\\
	 a + b -\log3\leq a^+ \leq a+b-1,\\
	\left(a - n\log 9 - \log C^*-\left(1 + \frac{2}{C^*9^n}\right)\right)\!\vee \log(1/\delta)  \leq b^+ \!\leq a\text.
	\end{gathered}
	\end{equation}
	The last inequality is obtained using \eqref{eq:process_C} and knowing $Z_n \leq \delta$.
	
	The upper bound derived in \lemref{lem:upper_bnd} yields
	\begin{equation}\label{eq:plus_iter_z}
	\begin{split}
	z^+ = z^2 + 2mp \leq & z^2 + 4(z+m)^{3/2}\\
	\leq & z^2 + 4(2z)^{3/2}\\
	\leq & 13z^{3/2}
	\end{split}
	\end{equation}
	\centerline{and we already have}
	\begin{equation}\label{eq:minus_iter_z}
	z^- \leq 2z\text.
	\end{equation}
	
	Now, define $\cG_{n,n_0}(\beta) $ as the event $\sum_{k=n_0}^n \mathbbm{1}_{\{S_k = +\}} \geq (n-n_0)\beta$, $\beta < 1/2$. For sufficiently large $n$, we know that $\cG_{n,n_0}(\beta)$ occurs with high probability as a result of the law of large numbers. This implies the existence of $n_1 \geq n_0$ satisfying $\bP(\cG_{n,n_0}(\beta)) \geq 1-\epsilon/3$, $n\geq n_1$. Note that $\bP(\cG_{n,n_0}(\beta) \cap \cB_{n,n_0}) \geq \gamma - 2\epsilon/3$ for $n \geq n_1$.

	Using the same machinery in \cite{Telatar_Rate}, one can refer to inequalities \eqref{eq:plus_iter_z}, \eqref{eq:minus_iter_z} and show that there exists an $n_2\geq n_0$ such that for any $\bs{s}_n \in \cG_{n,n_0}(\beta) \cap \cB_{n,n_0}$, both $(n \log 9) / B_n + \log  C^*(\delta,\epsilon) + \left(1 + \frac{2}{C^*9^n}\right)/B_n\leq 2^{-\alpha' n}$  and $\log(1/\delta)/B_n \leq 2^{-\alpha' n}$ for any $\alpha' < \log1.5/2$ and $n\geq n_2$.
	
		Define $R_n \triangleq A_n/B_n$. Again, from the upper bound in \lemref{lem:upper_bnd} one observes that $M_n \leq 2(4 Z_n)^{3/2}$. Thus $R_n = \frac{\log M_n}{\log Z_n} \geq \frac 3 2 + \frac{4}{\log Z_n} \geq \frac 3 2 + \frac{4}{\log \delta} \geq 1 + \delta(2+\delta)e^\delta$ for all $\bs{s}_n \in \cB_{n,n_0}(\delta)$ and for the previously chosen $\delta$.
	
	Referring to \eqref{eq:abupdates}, we have the following upper bound for $r^+$.
	\begin{equation*}
	r^+ \leq \frac{-1 /b + 1 + r}{\left(r - n\log 9/b - \log C^*/b-\left(1 + \frac{2}{C^*9^n}\right)/b\right)\!\vee \log(1/\delta)/b}
	\end{equation*}
	For $n \geq  n_3 \triangleq n_1 \vee n_2$ and same kind of $\bs{s}_n$, we know $r >1$, $n\log 9/b + \log C^*/b+\left(1 + \frac{2}{C^*9^n}\right)/b \leq 2^{-\alpha' n}$ and $\log(1/\delta)/b \leq 2^{-\alpha' n}$. Hence, the upper bound becomes
	$$
	r^+ \leq \frac{1 + r-2^{-\alpha' n}}{r-2^{-\alpha' n}}\text.
	$$
	In similar manner, iterations for $R_n$ are bounded as
	\begin{equation*}
	\begin{split}
	\frac{1 + r-2^{-\alpha' n}}{r} \leq r^+ \leq \frac{1 + r-2^{-\alpha' n}}{r-2^{-\alpha' n}}
	\end{split}
	\end{equation*}
	\centerline{and}
	$$
	r-2^{-\alpha' n}  \leq r^- \leq \frac{r}{1-2^{-\alpha' n}}\text.
	$$
	 From these, one concludes that
	\begin{equation}\label{eq:R_Plus_iteration}
	\left|R_{n+1}^+ - \left(\frac{R_n+1}{R_n}\right)\right|\leq 2^{-\alpha' n+1}
	\end{equation}
		\centerline{and}
		\begin{equation}\label{eq:R_M_iteration}
		R_n-2^{-\alpha' n+1} \leq R_{n+1}^-\leq R_n(1+2^{-\alpha' n+1}).
		\end{equation}
		Now, choose an $n_4$ such that $n_4 \triangleq \left\lceil \frac 1 {\alpha'} \log\left(\frac{2(2+\delta)e^\delta}{\delta(1-2^{-\alpha'})}\right)\right\rceil \vee n_3$. Define $\sigma_n \triangleq 2\sum_{k = n_4}^n2^{-\alpha' k}$ and observe $\sigma_n \leq  2\sum_{k = n_4}^\infty2^{-\alpha' k} \leq \delta$ for $n\geq n_4$. Since $R_n > 1$, the $+$ iteration ensures that $R_{n+1}^+ \leq 2+2^{-\alpha'n+1} \leq 2 + \delta$. Note that after exposed to $+$ transformation once, even infinitely many $-$ transformations cannot force $R_n$ to grow unboundedly as 
		$$
		R_\infty^{(-^\infty)} \leq (2+\delta)\prod_{k=n_4}^{\infty}(1+2^{-\alpha' k+1}) \leq(2+\delta)e^{\delta}.
		$$
	This shows that $R_n$ is bounded with probability close to $\gamma$. Using the upper bound found above, we obtain
	\begin{equation}\label{eq:R_Minus_iteration}
	|R_{n+1}^- - R_n| \leq  (2+\delta)e^\delta2^{-\alpha' n+1}\text.
	\end{equation}
 Define another process $X_n$ such that $X_{n_4} = R_{n_4}$ and 
	$$
	X_{n+1}^+ = \frac{X_n+1}{X_n},\quad X_{n+1}^- = X_n,\quad n\geq n_4\text.
	$$
 Using all these facts, we can also show that 
	
	\begin{equation}\label{ineq:sigma}
	|R_n - X_n| \leq (2+\delta)e^\delta\sigma_{n-1},\quad n \geq n_4\text.
	\end{equation}
	
	This follows by induction. The base case is easily proven from inequalities \eqref{eq:R_Plus_iteration} and \eqref{eq:R_Minus_iteration}. We now verify the other cases. Assuming the induction hypothesis we have $|R_n-X_n| \leq (2+\delta)e^\delta\sigma_{n-1}$.\\
	
	\noindent For $-$ iteration, we have
	$$
	|R^-_{n+1} - X^-_{n+1}| \leq |R_{n} - X_{n}| + (2+\delta)e^\delta2^{-\alpha' n+1} = (2+\delta)e^\delta\sigma_n\text.
	$$
	For $+$ iteration, we have
	\begin{align*}
	|R^+_{n+1} - X^+_{n+1}| &= \left|R^+_{n+1} - \frac{X_n+1}{X_n}\right|\\
	&\leq \left|\frac{R_n+1}{R_n} - \frac{X_n+1}{X_n}\right| + 2^{-\alpha' n+1}\text.
	\end{align*}
	We have assumed that $|R_n-X_n| \leq (2+\delta)e^\delta\sigma_{n-1}$. Note that since for all $n\geq n_4$, $\sigma_n \leq \delta$, this also implies $|R_n-X_n| \leq (2+\delta)e^\delta\delta$. Recall that $R_n \geq 1+\delta(2+\delta)e^\delta$ for all $\bs{s}_n \in \cB_{n,n_0}(\delta)$ and therefore $X_n \geq 1$ for such $\bs{s}_n$. The magnitude of derivative of $\frac{x+1}{x}$  is bounded by 1 on $[1,\infty)$. Hence,
	$$
	\left|\frac{R_n+1}{R_n} - \frac{X_n+1}{X_n}\right| \leq |R_n-X_n| \leq (2+\delta)e^\delta\sigma_{n-1}
	$$
	\centerline{and}
	$$
	|R^+_{n+1} - X^+_{n+1}| \leq (2+\delta)e^\delta\sigma_{n-1} + 2^{-\alpha' n+1} \leq (2+\delta)e^\delta\sigma_n\text.
	$$
	Therefore, we have proved the inequality \eqref{ineq:sigma} for all $n \geq n_4$. As we also have $(2+\delta)e^\delta\sigma_n \leq \delta$, we deduce
	
	\begin{equation}\label{eq:small1}
	|R_n - X_n| \leq  \delta,\quad n \geq n_4\text.
	\end{equation}
	
	Finally, we know that for any $\bs{s}_n \in \cG_{n,n_4}(\beta)$ and sufficiently large $n$, there will be arbitrarily large number of $+$ operations with high probability, say $\epsilon/3$. Since $|X_{n+1}^+ - \phi| = \left|\frac{X_n+1}{X_n}-\frac{\phi+1}{\phi}\right| = \left|\frac{X_n-\phi}{X_n\phi}\right| < \frac 1 \phi |X_n -\phi|$, $X_n$ converges to $\phi$. This shows the existence of an $n_5\geq n_4$ such that 
	\begin{equation}\label{eq:small2}
	|X_n - \phi| < \epsilon_r - \delta\text{ and } \bP(\cG_{n_5,n_4}(\beta)) \geq 1-\epsilon/3\text, \quad n \geq n_5\text.
	\end{equation}
	\eqref{eq:small1} and \eqref{eq:small2} imply $|R_n - \phi| \leq \epsilon_r$ for $n\geq n_5$ and $\bs{s}_n \in \cG_{n_4,n_0}(\beta) \cap \cB_{n,n_0} \cap \cG_{n,n_4}(\beta) = \cB_{n,n_0} \cap \cG_{n,n_0}(\beta) $ where 
	$$
	\bP(\cB_{n,n_0} \cap \cG_{n,n_0}(\beta)) \geq \gamma -\epsilon/3 -\epsilon/3 - \epsilon/3 = \gamma -\epsilon, \quad n \geq n_5\text.
	$$

\subsection{Proof of Theorem \ref{thm:rate}}\label{pf:rate}

We continue from the proof of Lemma \ref{lem:processR}. For all $\bs{s}_n \in \cB_{n,n_0} \cap \cG_{n,n_0}(\beta)$, $n\geq n_5$; we have $Z_n^{\phi+\epsilon_r} \leq M_n \leq Z_n^{\phi-\epsilon_r}$. Therefore, one obtains the following upper and lower bounds for iterations of $Z_n$.
\begin{align}\label{eq:phi}
\frac 1 3 z^{\phi+\epsilon_r} \leq z^+ \leq 3z^{\phi-\epsilon_r},\quad \frac 1 2z \leq z^- \leq 2z\text.
\end{align}
Let $\bar \beta \triangleq 1-\beta$. For a sufficiently large $n_6$, $\cG_{n_6,n_5}(\beta) \cap \cG_{n_6,n_5}(\bar\beta)^C$ occurs with high probability. For once again, the same machinery in \cite{Telatar_Rate} is used to obtain
\begin{equation}\label{eq:phi2}
2^{-2^{n\beta(\log\phi+\delta'')}} \leq Z_n \leq 2^{-2^{n\beta(\log\phi-\delta')}}
\end{equation}
for $n > n_6$, and any $\delta',\delta'' > 0$ which proves the first part of the theorem.

For the second part, the upper and lower bounds on iterations of $M_n$ are given by
\begin{align*}
\frac 1 3 m^{1+1/(\phi - \epsilon_r)}  \leq m_+ \leq 3m^{1+1/(\phi + \epsilon_r)},\quad \frac 1 2 m \leq m_- \leq 2m\text.
\end{align*}

Observe that $1+\frac{1}{\phi + \epsilon_r} \geq \phi - \epsilon_r$ and $1+\frac{1}{\phi - \epsilon_r} \leq \phi +\epsilon_r$ for small $\epsilon_r$.
Now, the same argument that we used to show \eqref{eq:phi2} from \eqref{eq:phi} allows us to conclude 

\begin{equation*}
2^{-2^{n\beta(\log\phi+\delta'')}} \leq M_n \leq 2^{-2^{n\beta(\log\phi-\delta')}}
\end{equation*}
from the bounds on $m^+$ and $m^-$.
\end{document}